\documentclass[a4paper,reqno]{amsart}
\usepackage{amsmath}
\usepackage{amssymb}
\usepackage{accents} 
\usepackage{todonotes}
\usepackage{enumitem}

\newtheorem{theorem}{Theorem}[section]

\newtheorem{lemma}[theorem]{Lemma}
\newtheorem{remark}[theorem]{Remark}
\newtheorem{example}[theorem]{Example}

\newcommand{\wt}{\widetilde}
\newcommand{\wh}{\widehat}
\newcommand{\ol}{\overline}
\numberwithin{equation}{section}
\newcommand{\orcidauthorNJ}{0000-0001-7504-4444}
\newcommand{\bX}{\boldsymbol{X}}
\newcommand{\bY}{\boldsymbol{Y}} 
\allowdisplaybreaks

\title{
QRT map on a bielliptic surface 
}
\author{Nalini Joshi}
\address{School of Mathematics and Statistics F07, The\,University\,of\, Sydney, NSW 2006, Australia.}
\email{nalini.joshi@sydney.edu.au}
\thanks{NJ's ORCID is \orcidauthorNJ}
\author{Frank W. Nijhoff}
\address{School of Mathematics, University of Leeds, Leeds LS2 9JT,
UK}
\email{frank.nijhoff@gmail.com}
\author{Allan Steel}
\address{School of Mathematics and Statistics F07, The\,University\,of\,Sydney, NSW 2006, Australia.}
\email{allan.steel@sydney.edu.au}
\subjclass[2020]{
  37J70, 
  14H70  
}

\begin{document}
\begin{abstract}
    The family of mappings of the plane possessing a biquadratic invariant, which is known collectively as QRT maps, is composed of two involutions, one preserving a vertical shift and the other preserving a horizontal shift in the plane. In this paper, we extend the map by replacing each shift by the group operation on each of two  families of elliptic curves, whose product forms a bielliptic surface.
\end{abstract}
\maketitle

\section{Introduction}

The Quispel-Roberts-Thompson (QRT) map is an 8-parameter family of 
mappings of the plane preserving a biquadratic invariant \cite{QRT}, cf. also \cite{HJN}. It is the composition of two involutions, which can be expressed as the addition formula on a rational elliptic surface \cite{Tsuda, Duist}.  The QRT map has played an important role in the development of the theory of 
discrete integrable systems, notably in the discovery of discrete 
Painlev\'e equations, see e.g., \cite[Chapter 11]{HJN} and  \cite{Joshi2019}. In this paper, we generalise the QRT map to bielliptic surfaces\footnote{Bielliptic surfaces are also known as hyperelliptic surfaces in earlier literature.}.

Let $\Gamma$ and $\Gamma^*$ denote two elliptic curves with respective uniformising variables $\xi$ and $\eta$, given by 
\begin{align}
 \label{eq:Weier}  
 \Gamma\ :\quad    X^2&=4x^3-g_2 x-g_3\ \, =4(x-e_1)(x-e_2)(x-e_3),\\
\label{eq:Weier2}
\Gamma^\ast:\quad   Y^2&=4y^3-G_2 y-G_3=4(y-E_1)(y-E_2)(y-E_3),
\end{align} 
with branch points given by coordinates $e_1,e_2,e_3$ and 
$E_1,E_2,E_3$ respectively, 
where we assume that, in general, the curves $\Gamma$ and $\Gamma^\ast$ are non-degenerate, i.e., 
\begin{align*}
    &27 g_3^2 \not= g_2^3,\quad 27 G_3^2 \not= G_2^3.
\end{align*} 
Furthermore, we do not impose any conditions on how the two curves $\Gamma$ and $\Gamma^\ast$ are related to each other; we leave the moduli 
$g_2,g_3, G_2,G_3$ undetermined (including their possible dependence on 
the variables of the other curve). 

In this paper, we define two involutions $\iota_j$, $j=1, 2$, one acting on $\Gamma$ and the other on $\Gamma^\ast$, by the group action on each respective curve. We will consider these involutions 
as a dynamical system on a bielliptic surface (see section \ref{s:geom} for the standard definition of a bielliptic surface). For this reason, we will call the composition of the two 
involutions the \textit{bielliptic QRT map}. In remark \ref{rem:QRT} we show how the invariant of the bielliptic QRT map degenerates to that of the standard QRT map in a singular limit of the elliptic curves.  

The involutions $\iota_j$ are introduced in section \ref{s:eqrt}, where we also show that 
their composition is symplectic. Our focus lies primarily on the case when $\iota_1$ acts on $\Gamma$ while $\Gamma^\ast$ is preserved and $\iota_2$ acts on $\Gamma^\ast$ while $\Gamma$ is preserved. Under the condition that $g_2$, $g_3$ are independent of $y$ and $G_2$, $G_3$ are independent of $y$, we provide the generic 10-parameter family of all such maps in Section \ref{s:consmoduli}, along with an explicit example.

\section{bielliptic QRT} \label{s:eqrt}

In this section, we construct the bielliptic QRT map, show that it is symplectic and briefly describe how it reduces to the standard QRT map in the case when the elliptic curves $\Gamma$ and $\Gamma^*$ become degenerate. 

Define an involution on the curve $\Gamma$ by imposing the determinantal equation
\begin{align}\label{eq:3det}
 \left| \begin{array}{ccc}
       1  & x & X  \\
       1  & \wt{x} & \wt{X} \\ 
       \phi_1 & \phi_2 & \phi_3 
    \end{array}\right| 
     = (\wt{x}-x)\phi_3 -(\wt{X}-X)\phi_2 +(x\wt{X}-\wt{x}X)\phi_1 = 0\ , \end{align} 
where we assume that the vector $\boldsymbol{\phi}=(\phi_1,\phi_2,\phi_3)$ is a constant (i.e., 
independent of $(x,X)$). The involution is given by the map 
\begin{equation}\label{eq:i_1}
\iota_1:\quad (x,X) \mapsto (\wt{x},\wt{X}),
\end{equation}
where $(x,X)$ denotes points on the Weierstrass curve $\Gamma$. For given 
$\boldsymbol{\phi}=(\phi_1,\phi_2,\phi_3)$, we can solve equation \eqref{eq:3det} for 
$\wt{x}$, and substituting the result (in terms of $x$, $X$ and $\wt{X}$) into the curve relation for $\Gamma$ we 
obtain a cubic for $\wt{X}$, which has three solutions, one of which is $\wt{X}=X$ (the trivial map) 
and the other is $\wt{X}=\phi_3/\phi_1$, provided $(\phi_2/\phi_1,\phi_3/\phi_1)$ lies on the 
elliptic curve $\Gamma$ as well.  The condition 
$(\phi_2/\phi_1,\phi_3/\phi_1)\in\Gamma$ holds without loss of generality because 
neither the moduli $g_2$ and $g_3$ nor the corresponding uniformising variable 
are specified. The remaining solution of the cubic is the one we are 
interested in, and, as shown in the following lemma, it defines a single-valued involution $\iota_1$. 
\begin{lemma}\label{lem:xinvolution}
The map defined by \eqref{eq:3det} together with \eqref{eq:Weier} is an involution, i.e., $\wt{\wt{\boldsymbol{X}}}=\boldsymbol{X}$. 
\end{lemma}
\begin{proof}
To prove the statement, apply the map $\iota_1$ to the determinant on the left side of Equation \eqref{eq:3det}. Then, in addition to \eqref{eq:3det}, we obtain 
\begin{align}
    \left| \begin{array}{ccc}
       1  & \wt{\wt{x}} & \wt{\wt{X}}  \\
       1  & \wt{x} & \wt{X} \\ 
       \phi_1 & \phi_2 & \phi_3 
    \end{array}\right|=0 \ , 
\end{align}
given that the components of $\boldsymbol{\phi}$ are constant. Subtracting \eqref{eq:3det} from this determinant, we see that a solution is given by $\wt{\wt{x}}=x,~ \wt{\wt{X}}=X$. However, there are other solutions. So, assuming that $\phi_1$ does not identically vanish, and parametrising the curve in terms of Weierstrass elliptic functions, i.e., setting $(x,X)=(\wp(\xi),\wp'(\xi))$ for some dynamical variable $\xi$, and given (by the argument presented above) that $\phi_1,\phi_2,\phi_3$ can be parametrised such that $(\frac{\phi_2}{\phi_1},\frac{\phi_3}{\phi_1})=(\wp(\delta),\wp'(\delta))$, we obtain from \eqref{eq:3det} the Frobenius-Stickelberger determinant, cf. \cite{FS},
\begin{align}
      \left| \begin{array}{ccc}
       1  & \wp(\xi) & \wp'(\xi) \\
       1  & \wp(\wt{\xi}) & \wp'(\wt{\xi}) \\ 
       1 & \wp(\delta) & \wp'(\delta) 
    \end{array}\right|= 2\frac{\sigma(\xi+\wt{\xi}+\delta)}{\sigma^3(\xi)\,\sigma^3(\wt{\xi})\,
    \sigma^3(\delta)}\,\sigma(\wt{\xi}-\xi)\,\sigma(\wt{\xi}-\delta)\,\sigma(\xi-\delta)\ , \label{eq:Frob} 
\end{align}
and this vanishes if either $\wt{\xi}=\xi$ (modulo the period lattice of the Weierstrass $\sigma$-function), which is a trivial solution, or if $\wt{\xi}+\xi+\delta=0$ (modulo the period lattice).  But the latter implies $\wt{\wt{\xi}}=\xi$, and hence the map $\boldsymbol{X}\rightarrow \wt{\boldsymbol{X}}$ is an involution. 
\end{proof}
\begin{remark}
Note that the proof of the 
lemma implies that 
\begin{equation}\label{eq:fcubic}
\left(\dfrac{\phi_3}{\phi_1}\right)^2=4\left(\dfrac{\phi_2}{\phi_1}\right)^3-g_2 \left(\dfrac{\phi_2}{\phi_1}\right)-g_3.
\end{equation}
\end{remark}

Inspired by the QRT construction we now proceed by constructing an invariant for the map $\iota_1$. 
Introducing the vector $\boldsymbol{X}=(1,x,X)^T$, we can write \eqref{eq:3det} as 
\begin{align}\label{eq:phiXX}
    \boldsymbol{\phi}\cdot (\boldsymbol{X}\times \wt{\boldsymbol{X}})=0\ . 
\end{align}  
Assuming that $\boldsymbol{\phi}$ depends on another set of variables $y$ and $Y$ such that 
the vector $\boldsymbol{Y}=(1,y,Y)^T$ is constant under the map: $\boldsymbol{Y}\rightarrow 
\wt{\boldsymbol{Y}}=\boldsymbol{Y}$, and taking $\boldsymbol{\phi}$ of the form 
\begin{align}\label{eq:f} 
\boldsymbol{\phi}=\boldsymbol{A}_0\boldsymbol{Y}\times \boldsymbol{A}_1\boldsymbol{Y}\ , 
\end{align}  
we find that the ratio 
\begin{align}\label{eq:inv} 
I=\frac{\boldsymbol{X}^T\boldsymbol{A}_0\boldsymbol{Y}}{\boldsymbol{X}^T\boldsymbol{A}_1\boldsymbol{Y}}
\end{align}  
is an invariant of the map $(x,X;y,Y)\rightarrow(\wt{x},\wt{X};\wt{y},\wt{Y})$, i.e. $\wt{I}=I$. This follows directly from inserting equation \eqref{eq:f} into \eqref{eq:phiXX} leading to 
\begin{align*} 
(\boldsymbol{A}_0\boldsymbol{Y}\times \boldsymbol{A}_1\boldsymbol{Y})\cdot (\boldsymbol{X}\times \wt{\boldsymbol{X}})= 
(\boldsymbol{X}^T\boldsymbol{A}_0\boldsymbol{Y})
(\wt{\boldsymbol{X}}^T\boldsymbol{A}_1\boldsymbol{Y}) 
- (\boldsymbol{X}^T\boldsymbol{A}_1\boldsymbol{Y})
(\wt{\boldsymbol{X}}^T\boldsymbol{A}_0\boldsymbol{Y})=0 ,  
\end{align*} 
and noting that $\wt{\boldsymbol{Y}}=\boldsymbol{Y}$. 
Importantly, we can allow for the moduli of the curve $\Gamma$ to depend on the variables $y$ and $Y$, 
i.e., $g_2=g_2(\boldsymbol{Y})$ and $g_3=g_3(\boldsymbol{Y})$. 

Since the map $\iota_1$ is an involution, i.e., a map of order 2, we need to extend it to create a map 
of infinite order. This is done by introducing a second involution $\iota_2$, which does not 
commute with $\iota_1$, namely 
\begin{equation}\label{eq:i2}
\iota_2:\quad (y,Y) \mapsto (\wh{y},\wh{Y}),
\end{equation}
which leaves the variables $x,X$ invariant, $\wh{\boldsymbol{X}}=\boldsymbol{X}$, and such that 
\begin{align}\label{eq:3dett}
& \boldsymbol{\psi}\cdot (\boldsymbol{Y}\times \wh{\boldsymbol{Y}}) \nonumber \\ 
& = \left| \begin{array}{ccc}
       1  & y & Y  \\
       1  & \wh{y} & \wh{Y} \\ 
       \psi_1 & \psi_2 & \psi_3 
    \end{array}\right| \\
    & = (\wh{y}-y)\psi_3 -(\wh{Y}-Y)\psi_2 +(y\wh{Y}-\wh{y}Y)\psi_1 = 0\ . \nonumber 
\end{align}
This map shares the same invariant $I$ of \eqref{eq:inv} provided we take the vector 
$\boldsymbol{\psi}=(\psi_1,\psi_2,\psi_3)$ of the form 
\begin{align}\label{eq:g} 
\boldsymbol{\psi}=\boldsymbol{A}_0^T\boldsymbol{X}\times \boldsymbol{A}^T_1\boldsymbol{X}\ , 
\end{align}  
and where we assume that $(y,Y)$ are points of the second Weierstrass curve $\Gamma^\ast$, 
with moduli $G_2=G_2(\boldsymbol{X})$ and $G_3=G_3(\boldsymbol{X})$ which 
may depend on the variables $x,X$. Similarly to equation \eqref{eq:fcubic}, the functions 
$\psi_i(X)$ are subject to the curve relation on $\Gamma^\ast$, namely
\begin{equation}\label{eq:gcubic}
\left(\dfrac{\psi_3}{\psi_1}\right)^2=4\left(\dfrac{\psi_2}{\psi_1}\right)^3-G_2 \left(\dfrac{\psi_2}{\psi_1}\right)-G_3.
\end{equation}

The composition of the two non-commuting involutions $\iota_1: (x,X;y,Y)\to (\wt{x},\wt{X};\wt{y},\wt{Y})$ and 
$\iota_2: (x,X;y,Y)\to (\wh{x},\wh{X};\wh{y},\wh{Y})$, viewed as maps on the space of four variables 
$(x,X;y,Y)$, leads to 
a map $\mu=\iota_2\circ\iota_1$. More explicitly we have the chain 
\begin{align}
     (x,X;y,Y)~\stackrel{\iota_1}{\to}~ (\wt{x},\wt{X};y=\wt{y},Y=\wt{Y})~\stackrel{\iota_2}{\to}~
      (\wt{x}=\wh{\wt{x}},\wt{X}=\wh{\wt{X}};\wh{\wt{y}},\wh{\wt{Y}})
\end{align}
resulting in a combined map 
\begin{subequations}\label{eq:ellQRT} 
\begin{align}\label{eq:mu} 
    \mu: ~  (x,X;y,Y)~\to~ (\overline{x},\overline{X};\overline{y},\overline{Y})
\end{align}
given by
\begin{align}
   &  \left(\tfrac{1}{2}\phi_1(\ol{x}+x)-\phi_2\right)\,\frac{\ol{X}-X}{\ol{x}-x}=\tfrac{1}{2}\phi_1(\ol{X}+X)-\phi_3\ , \\ 
    &  \left(\tfrac{1}{2}\ol{\psi}_1(\ol{y}+y)-\ol{\psi}_2\right)\,\frac{\ol{Y}-Y}{\ol{y}-y}=\tfrac{1}{2}\ol{\psi}_1(\ol{Y}+Y)-\ol{\psi}_3\ , 
\end{align}
subject to 
\begin{align}\label{eq:curves}
X^2=4x^3-g_2(Y)x-g_3(Y)\ , \quad Y^2=4y^3- G_2(X) y-G_3(X) , 
\end{align} 
and where $\boldsymbol{\phi}=\boldsymbol{\phi}(\boldsymbol{Y})$ and $\ol{\boldsymbol{\psi}}=\boldsymbol{\psi}(\ol{\boldsymbol{X}})$ are given by \eqref{eq:f} and \eqref{eq:g} respectively.
\end{subequations}

The map $\mu$ defined by \eqref{eq:mu} is what we will refer to as the bielliptic QRT map. 
This map is symplectic through the following argument (in contrast to the QRT map which in general 
position is only measure-preserving, cf. \cite{Roberts-thesis}, also see \cite[\S 6.3.3]{HJN}). 
In fact, in terms of the uniformising variables on the curves $\Gamma$ and $\Gamma^\ast$,
the proof of measure-preservation is straightforward. Recalling that the proof of Lemma \ref{lem:xinvolution} gives the realisation of the 
involution $\iota_1: (\xi,\eta)\to (\wt{\xi},\wt{\eta})=(-\xi-\delta,\eta)$, we have 
\[ 
J_1=\frac{\partial(\wt{\xi},\wt{\eta})}{\partial(\xi,\eta)}
=\left|\begin{array}{cc} \dfrac{\partial\wt{\xi}}{\partial\xi} & \ast \\ 
  0 & 1 \end{array}\right|=-1 . 
\] 
Similarly, we have for $\iota_2:(\xi,\eta)\to (\wh{\xi},\wh{\eta})=(\xi,-\eta-\varepsilon)$, 
\[ 
J_2=\dfrac{\partial(\wh{\xi},\wh{\eta})}{\partial(\xi,\eta)}=\left|\begin{array}{cc} 1 & 0 \\ 
\ast & \frac{\partial\wh{\eta}}{\partial\eta} \end{array}\right|=-1 .  
\] 
 
 Hence, the Jacobian of the composition map $\iota_2\circ\iota_1$ 
 is given by $J=\wt{J}_2J_1=1$. which implies that in terms of the uniformising variables 
the map is not only measure-preserving, but actually symplectic. 
In terms of the variables on the bielliptic surface $S$, with affine coordinates $(x, X, y, Y)\in\mathbb C^4$, we define the symplectic form to be
\[ \Omega=\frac{{\rm d}x\wedge{\rm d}y}{X Y}\ . \] 
As a consequence of the derivation above this symplectic form is preserved under both involutions $\iota_1$ and $\iota_2$, i.e., $\wt{\Omega}=\wh{\Omega}=\Omega$. 

The following remark shows that the invariant of the bielliptic QRT map contains the one of the standard QRT map, when the curves $\Gamma$ and $\Gamma^*$ become degenerate.
\begin{remark}\label{rem:QRT} 
Consider the following degeneration of the curve $\Gamma$ {\rm (}respectively $\Gamma^*${\rm )}, which occurs when a pair of its branch points coalesces, with the other pair separating far apart. 

First, consider the change of variables $x=e_3+u^2$, $X=U\,u$, which transforms $\Gamma$ into a curve of the form 
\[ U^2=4(u^2-\gamma_0^2)(u^2-\gamma_1^2), \]
where $\gamma_0^2=(e_1-e_3)$ and $\gamma_1^2=(e_2-e_3)$. Assuming the dilation of one pair of branch points and coalescence of the remaining pair by taking $\gamma_0^2=-1/\epsilon^2$, $\gamma_1^2=\mu^2\,\epsilon^2$, for sufficiently small $\epsilon$, the resulting equation for $(u, U)$ becomes
\[
U^2=\dfrac{4}{\epsilon^2}\,u^2\bigl(1 + \mathcal O(\epsilon^2)\bigr),
\]
i.e., $U \sim \pm 2 u/\epsilon$. Putting these steps together leads to the transformation $x=e_3+ u^2$, $U= 2\,u^2/\epsilon$. The transformation of $(y, Y)$ on $\Gamma^*$ to $(v, V)$ proceeds in a similar manner under the corresponding assumptions on $E_1-E_3$ and $E_2-E_3$.

Substituting these transformations into the invariant $I$, while assuming $c_0=c_0^\circ\epsilon$, $c_1=c_1^\circ\epsilon$, $f_0=f_0^\circ\epsilon$, $f_1=f_1^\circ\epsilon$, $g_0=g_0^\circ\epsilon$, $g_1=g_1^\circ\epsilon$, $h_0=h_0^\circ\epsilon$, $h_1=h_1^\circ\epsilon$, and $i_0=i_0^\circ\epsilon^2$, $i_1=i_1^\circ\epsilon^2$, leads to a new invariant $I^\circ$ that is independent of $\epsilon$ as $\epsilon\to0$. This invariant is biquadratic in two variables $u$ and 
$v$ and, so, we obtain the generic invariant of the standard QRT map. 
\end{remark}

\section{Constant moduli case}\label{s:consmoduli}

The components of $\boldsymbol{\phi}=(\phi_1,\phi_2,\phi_3)$ defined by equation \eqref{eq:f} satisfy the cubic equation \eqref{eq:fcubic} and correspondingly the components of $\boldsymbol{\psi}=(\psi_1,\psi_2,\psi_3)$ satisfy equation \eqref{eq:gcubic}. These conditions place constraints on the entries of the matrices 
$A_j$.  In this section, we provide the general solution of these constraints, under the condition that the moduli $g_2(\bY)$, $g_3(\bY)$ as well as $G_2(\bX)$, $G_3(\bX)$
are constants, i.e., independent of $\bX$ and $\bY$. Although the  
option that the moduli may depend in a nontrivial way on the 
variables of the map provides us with possibly additional freedom in constructing interesting 
dynamical systems, we will not consider this situation in the current paper, and restrict ourselves here to the constant moduli case. 

The statement of finding the general solution for the latter case  
assumes that certain factors appearing in the denominators do not vanish. 
Their vanishing leads to 4 additional special solution branches, which are 
listed in Appendix \ref{app:special}.
In what follows we write the matrices $A_0$ and $A_1$ in terms of the 
following entries: 
\[
A_j=\begin{pmatrix}
a_j&b_j&c_j\\
d_j&e_j&f_j\\
g_j&h_j&i_j\\
\end{pmatrix}\ , \quad j=0,1\ .
\]
(Note that the (3,1)-entry of each of $A_0$, $A_1$ is $g_0$ or $g_1$ respectively and should not be confused with $g_2$, $g_3$ in the definition of $\Gamma$.) 
Taking the components of Equation \eqref{eq:f} and substituting into \eqref{eq:fcubic}, we get an equation in which we can use \eqref{eq:Weier2} to replace even powers of $Y$. The remaining equation is linear in $Y$. Since each coefficient in $Y$ must vanish identically, we obtain two equations that are polynomial in $y$. The vanishing of each coefficient of a power of $y$ leads to 18 polynomial equations for the entries of $A_0$ and $A_1$. 

The generic solution of these equations was found by using a variant
of the primary decomposition algorithm in the {\sc Magma} Computer
Algebra System \cite{BCP}.  The algorithm computes a decomposition of the
radical of the ideal generated by the polynomials (following the general
approach of \cite{GKZ}, combined with sparse interpolation techniques
from \cite{Zippel}).

Each component of the decomposition gives a branch of the complete
solution to the equations, and is described by a list of free parameters,
a list of main variables $v_1, \ldots, v_l$ and a triangular system of
polynomials in the $v_i$, so that the solutions for the branch can be
obtained by setting the parameters to arbitrary values and then solving
for the $v_i$ in the triangular system by back-substitution, assuming
the relevant denominators are non-zero.

The main branch of the solution is described with the free parameters
$b_0$, $e_0$, $f_0$, $h_0$, $i_0$, $e_1$, $f_1$, $g_1$, $h_1$, $i_1$,
$g_2$, $g_3$, $G_2$, $G_3$ and the main variables $a_0$, $c_0$, $d_0$,
$g_0$, $a_1$, $b_1$, $c_1$, $d_1$.  Write $D_0=e_0 i_1 - h_0 f_1$,
$D_1=e_1 i_1 - f_1 h_1$.  Then if $D_0\not=0$, $D_1\not=0$, $h_0\not=0$,
$i_1\not=0$, the solution is described by the following triangular system
in the main variables\footnote{In the list \eqref{eq:conditions2} the choice of the parameters 
for which we solve the constraint equations \eqref{eq:fcubic} and \eqref{eq:gcubic} are 
determined by the monomial ordering of the Gr\"obner basis used in the algorithm.}:

\begin{subequations}\label{eq:conditions2}
\begin{align}
a_0 &= \dfrac{1}{D_0 D_1} \Bigl((-e_0 i_0 h_1 +
f_0 h_0 h_1) c_1 d_1 + (e_0 i_0 e_1 g_1 - f_0 h_0 e_1 g_1) c_1\nonumber \\
& \quad\qquad + (b_0 e_0 i_1^2 - b_0 f_0 h_1 i_1 -
b_0 h_0 f_1 i_1 + b_0 i_0 f_1 h_1) d_1 \nonumber\\
& \quad\qquad - b_0 e_0 f_1 g_1 i_1 + b_0 f_0 e_1 g_1 i_1 + b_0 h_0 f_1^2 g_1
- b_0 i_0 e_1 f_1 g_1\Bigr), \\
c_0 &= \dfrac{1}{D_0} \bigl(\bigl(e_0 i_0 - f_0 h_0\bigr) c_1 + b_0 f_0 i_1 - b_0 i_0 f_1\bigr), \\
d_0 &= \dfrac{1}{D_1} \bigl(\bigl(e_0 i_1 - f_0 h_1\bigr) d_1 - e_0 f_1 g_1 + f_0 e_1 g_1\bigr), \\
g_0 &= \dfrac{1}{D_1} \bigl(\bigl(h_0 i_1 - i_0 h_1\bigr) d_1 - h_0 f_1 g_1 + i_0 e_1 g_1\bigr), \\
a_1 &= \dfrac{1}{D_0} \bigl(\bigl(-h_0 d_1 + e_0 g_1\bigr) c_1 + b_0 i_1 d_1 - b_0 f_1 g_1\bigr),\\
b_1 &= \dfrac{1}{D_0} \bigl(\bigl(e_0 h_1 - h_0 e_1\bigr) c_1 + b_0 e_1 i_1 - b_0 f_1 h_1\bigr),\\
c_1^3 &=
\dfrac{1}{h_0^3}
    \Bigl(\bigl(3 b_0 h_0^2 i_1 + \dfrac{1}{4} e_0^3 i_1 - \dfrac{1}{4} e_0^2 h_0 f_1\bigr) c_1^2\nonumber\\
&\quad\qquad +
\bigl(-3 b_0^2 h_0 i_1^2 - \dfrac{1}{2} b_0 e_0^2 f_1 i_1 + \dfrac{1}{2} b_0 e_0 h_0 f_1^2 + \nonumber\\
&\quad\qquad\qquad\dfrac{1}{4} e_0^2 h_0 i_1^2 g_2 -
\dfrac{1}{2} e_0 h_0^2 f_1 i_1 g_2 + \dfrac{1}{4} h_0^3 f_1^2 g_2\bigr)c_1 \nonumber\\
&\quad\qquad + \bigl(b_0^3 i_1^3 + \dfrac{1}{4} b_0^2 e_0 f_1^2 i_1 -
\dfrac{1}{4} b_0^2 h_0 f_1^3 - \dfrac{1}{4} b_0 e_0^2 i_1^3 g_2 +
\nonumber\\
&\quad\qquad\qquad
\dfrac{1}{2} b_0 e_0 h_0 f_1 i_1^2 g_2 -
\dfrac{1}{4} b_0 h_0^2 f_1^2 i_1 g_2 + \dfrac{1}{4} e_0^3 i_1^3 g_3
\nonumber\\
&\quad\qquad\qquad
- \dfrac{3}{4} e_0^2 h_0 f_1 i_1^2 g_3 +
\dfrac{3}{4} e_0 h_0^2 f_1^2 i_1 g_3 - \dfrac{1}{4} h_0^3 f_1^3 g_3\bigr)\Bigr),\\
d_1^3 &=
\dfrac{1}{i_1^3}
    \Bigl(
    \bigl(-\dfrac{1}{4} e_1 h_1^2 i_1 + 3 f_1 g_1 i_1^2 + \dfrac{1}{4} f_1 h_1^3\bigr) d_1^2\nonumber\\
&\quad\qquad + \bigl(\dfrac{1}{2} e_1^2 g_1 h_1 i_1 +
    \dfrac{1}{4} e_1^2 i_1^3 G_2 - \dfrac{1}{2} e_1 f_1 g_1 h_1^2 \nonumber\\
&\quad\qquad\qquad - \dfrac{1}{2} e_1 f_1 h_1 i_1^2 G_2
    - 3 f_1^2 g_1^2 i_1 + \dfrac{1}{4} f_1^2 h_1^2 i_1 G_2\bigr) d_1\nonumber\\
&\quad\qquad + (- \dfrac{1}{4} e_1^3 g_1^2 i_1 - \dfrac{1}{4} e_1^3 i_1^3 G_3
    + \dfrac{1}{4} e_1^2 f_1 g_1^2 h_1
    - \dfrac{1}{4} e_1^2 f_1 g_1 i_1^2 G_2 \nonumber\\
&\quad\qquad\qquad + \dfrac{3}{4} e_1^2 f_1 h_1 i_1^2 G_3 +
    \dfrac{1}{2} e_1 f_1^2 g_1 h_1 i_1 G_2 -
    \dfrac{3}{4} e_1 f_1^2 h_1^2 i_1 G_3\nonumber\\
&
\quad\qquad\qquad
+ f_1^3 g_1^3 - \dfrac{1}{4} f_1^3 g_1 h_1^2 G_2 + \dfrac{1}{4} f_1^3 h_1^3 G_3)\Bigr).
\end{align}
\end{subequations}

\begin{example}\label{ex:integer_ex} The integer matrices
\[
A_0=\begin{pmatrix}
2&-16&2\\
3&-2&3\\
-2&5&-2\\
\end{pmatrix}, \quad
A_1=\begin{pmatrix}
-8&-2&-8\\
2&1&2\\
1&0&1\\
\end{pmatrix},
\]
provide solutions of equations \eqref{eq:conditions2}, under the conditions
\[
2g_2 + g_3 = 16, \quad G_3=-1.
\]
\end{example}
There are 4 main other branches of the solution listed in Appendix \ref{app:special} for the cases where at
least one of $D_0$, $D_1$, $h_0$ or $i_0$ is zero .

\section{Geometry of the bielliptic QRT map}\label{s:geom}
To make a connection to the standard terminology in the literature, we recall the construction of a bielliptic surface $S$ of type $I_a$, as explained in \cite{GriffithsHarris}. Assuming that the period lattice of $\Gamma^*$ is generated by $\{1, \tau\}$, we note that $\Gamma^*$ has an automorphism of order 2, given by $\eta\mapsto -\eta$, or equivalently $Y\mapsto -Y$. The surface $S$ is generated by quotienting $\Gamma\times \Gamma^*$ by the group of automorphisms generated by 
\[ \phi(\xi, \eta)=\left(\xi+\dfrac{\tau}{2}, -\eta\right).\]
While we do not make use of this construction explicitly, we rely on the fact that the resulting bielliptic surface $S$ is a smooth algebraic surface.  

In contrast, the geometric construction of the standard QRT map shows that it is a dynamical system on an elliptic surface \cite{Tsuda,Duist}. The bielliptic QRT map however can be considered as a higher dimensional 
generalisation in 4D space of variables $(x,X;y,Y)\in 
\mathbb{P}^2\times \mathbb{P}^2$ on the 
pencil 
\[ \{ (x,X;y,Y)\,|\,\boldsymbol{X}^T\boldsymbol{A}_0\boldsymbol{Y}-\lambda 
\boldsymbol{X}^T\boldsymbol{A}_1\boldsymbol{Y}=0 \}_{\lambda\in \mathbb{P}^1}\ .  \]
Rather than intersecting the corresponding surface 
with lines, as in \cite{Tsuda}, we intersect with 
the elliptic curves $\Gamma$ and $\Gamma^\ast$. In fact, the pencil can be viewed 
as the solution space of a 10-parameter family of affine-linear equations: 
\begin{equation}\label{eq:affinelinear}
Q_\lambda(x,X,y,Y;\boldsymbol{A}_0,\boldsymbol{A}_1)=0 \ , 
\end{equation}
where the bielliptic QRT map comes from a reduction of the solution space to solutions coming from the intersection with the curves $\Gamma$ and 
$\Gamma^\ast$. 

Equation \eqref{eq:affinelinear} resembles quadrilateral lattice equations (i.e., partial difference equations on a 2-dimensional grid), which are integrable in the sense that they satisfy the 
condition of consistency-around-the-cube (CAC) and were classified in \cite{ABS}. This possible connection between the two subjects is tantalising. In particular, we note that the most general equation in that class, known as Q4, has coefficients 
parametrised by an elliptic curve and possesses a class of $N$-soliton type solutions involving a pair of elliptic curves \cite{AtkNij2010}.

\section{Conclusion}
In this paper, we described an integrable, symplectic dynamical system on a bielliptic surface, which generalises the QRT map. The situation we focused  on is that 
where the moduli of the two curves $\Gamma$ and $\Gamma^\ast$ are constant (i.e., 
independent of the variables of the companion curve), in which case we find an 10-parameter family of two-dimensional symplectic maps possessing an algebraic invariant. The alternative option 
is to allow the moduli to depend on the companion variables, in which case we would have 
a more complicated system on a bielliptic surface. 
We will explore the latter situation in a future publication. 

There are many remaining open questions. 
One concerns special cases of matrices $A_j$, such as  
diagonal, symmetric, skew symmetric or upper/lower triangular matrices. These may be 
related to the classification of periodic fix points of the bielliptic QRT map, 
in the spirit of \cite{Tsuda}. 
Another concerns the issue of what happens when the Weierstrass curves $\Gamma$ and $\Gamma^\ast$ are birationally transformed to elliptic curves in general position, 
and whether that would lift some of the conditions \eqref{eq:conditions2} on the 
matrices $A_i$. The geometric classification of bielliptic surfaces also raises interesting questions on how the bielliptic QRT map may be described geometrically.

\section*{Acknowledgments}
FWN is grateful for the support and hospitality of the Sydney Mathematical Research Institute (SMRI) during a visit where the current work was initiated. He was supported by the Foreign Expert Program of the Ministry of Sciences and Technology of China, grant number G2023013065L when this work was finalised. 

\appendix
\section{Special branches of solutions}\label{app:special}
In this appendix, we list the special solutions that arise when one of $i_1$, $e_1 i_1 - f_1 h_1$, $h_0$, $e_0 i_1 - h_0 f_1$ vanishes. 
\subsection{Case 1: $i_1 = 0$}
\subsubsection{Case 1.1: $i_1 = 0$, $h_0h_1f_1\not=0$}
\begin{align*}
    a_0 &= -\frac{(e_0 i_0 - f_0 h_0)}{h_0 f_1^2} c_1 d_1 - \frac{(-e_0 i_0 e_1 g_1 + f_0 h_0 e_1 g_1)}{h_0 f_1^2 h_1} c_1 \\
    &\qquad + \frac{b_0 i_0}{h_0 f_1} d_1 - \frac{(-b_0 h_0 f_1 g_1 + b_0 i_0 e_1 g_1)}{h_0 f_1 h_1}, \\
    c_0 &= -\frac{(e_0 i_0 - f_0 h_0)}{h_0 f_1} c_1 + \frac{b_0 i_0}{h_0}, \\
    d_0 &= \frac{f_0}{f_1} d_1 - \frac{(-e_0 f_1 g_1 + f_0 e_1 g_1)}{f_1 h_1}, \\
    g_0 &= \frac{i_0}{f_1} d_1 - \frac{(-h_0 f_1 g_1 + i_0 e_1 g_1)}{f_1 h_1}, \\
    a_1 &= \frac{1}{f_1} c_1 d_1 - \frac{e_0 g_1}{h_0 f_1} c_1 + \frac{b_0 g_1}{h_0}, \\
    b_1 &= -\frac{(e_0 h_1 - h_0 e_1)}{h_0 f_1} c_1 + \frac{b_0 h_1}{h_0}, \\
    c_1^3 &= -\frac{1}{4} \frac{e_0^2 f_1}{h_0^2} c_1^2 - \frac{(-\frac{1}{2} b_0 e_0 f_1^2 - \frac{1}{4} h_0^2 f_1^2 g_2)}{h_0^2} c_1 - \frac{(\frac{1}{4} b_0^2 f_1^3 + \frac{1}{4} h_0^2 f_1^3 g_3)}{h_0^2}, \\
    d_1^2 &= \frac{2 e_1 g_1}{h_1} d_1 - \frac{(e_1^2 g_1^2 h_1 + 4 f_1^2 g_1^3 - f_1^2 g_1 h_1^2 G_2 + f_1^2 h_1^3 G_3)}{h_1^3}.
\end{align*}

\subsubsection{Case 1.2: $i_1 = 0$, $h_0=0$, with $e_0h_1f_1\not=0$}
\begin{align*}
    a_0 &= \frac{i_0}{f_1 h_1} b_1 d_1 - \frac{i_0 e_1 g_1}{f_1 h_1^2} b_1 + \frac{b_0 f_0 h_1 - b_0 i_0 e_1}{e_0 f_1 h_1} d_1 \\
    &\qquad + \frac{b_0 e_0 f_1 g_1 h_1 - b_0 f_0 e_1 g_1 h_1 + b_0 i_0 e_1^2 g_1}{e_0 f_1 h_1^2}, \\
    c_0 &= \frac{i_0}{h_1} b_1 - \frac{b_0 f_0 h_1 - b_0 i_0 e_1}{e_0 h_1}, \\
    d_0 &= \frac{f_0}{f_1} d_1 - \frac{-e_0 f_1 g_1 + f_0 e_1 g_1}{f_1 h_1}, \\
    g_0 &= \frac{i_0}{f_1} d_1 - \frac{i_0 e_1 g_1}{f_1 h_1}, \\
    a_1 &= \frac{g_1}{h_1} b_1 + \frac{b_0}{e_0} d_1 - \frac{b_0 e_1 g_1}{e_0 h_1}, \\
    b_1^2 &= \frac{2 b_0 e_1}{e_0} b_1 - \frac{4 b_0^3 h_1^2 + b_0^2 e_0 e_1^2 - b_0 e_0^2 h_1^2 g_2 + e_0^3 h_1^2 g_3}{e_0^3}, \\
    c_1 &= \frac{b_0 f_1}{e_0}, \\
    d_1^2 &= \frac{2 e_1 g_1}{h_1} d_1 - \frac{e_1^2 g_1^2 h_1 + 4 f_1^2 g_1^3 - f_1^2 g_1 h_1^2 G_2 + f_1^2 h_1^3 G_3}{h_1^3}.
\end{align*}
\subsubsection{Case 1.3: $i_1 = 0$,$h_0=0$, $f_1=0$, with $i_0\,h_1\,e_0\not=0$}
\begin{align*}
    a_0 &= \frac{g_0 b_1}{h_1}  + b_0 \frac{( f_0 h_1 - i_0 e_1)}{e_0 h_1 i_0} g_0 + \frac{b_0 g_1}{h_1}, \\
    c_0 &= \frac{i_0 b_1}{h_1} + b_0 \frac{( f_0 h_1 - i_0 e_1)}{e_0 h_1}, \\
    d_0 &= \frac{f_0 g_0}{i_0} + \frac{e_0 g_1}{h_1}, \\
    g_0^2 &= -\frac{i_0^2}{h_1^3} \Bigl(h_1^3 G_3 - h_1^2 g_1 G_2 + 4 g_1^3\Bigr), \\
    a_1 &= \frac{g_1 b_1}{h_1} , \\
    b_1^2 &= \frac{2 b_0 e_1}{e_0} b_1 - \frac{4 b_0^3 h_1^2 + b_0^2 e_0 e_1^2 - b_0 e_0^2 h_1^2 g_2 + e_0^3 h_1^2 g_3}{e_0^3}, \\
    c_1&=0, \\
    d_1 &= \frac{e_1 g_1}{h_1}.
\end{align*}
\subsubsection{Case 1.4} $i_1 = 0$,$h_0=0$, $f_1=0$, $i_0=0$, with $h_1 e_0\not=0$
\begin{align*}
    a_0 &= \frac{b_0}{e_0} d_0, \\
    c_0 &= \frac{b_0 f_0}{e_0}, \\
    d_0^2 &= 2 \frac{e_0 g_1}{h_1} d_0 - \frac{e_0^2 h_1 g_1^2 + f_0^2 h_1^3 G_3 - f_0^2 h_1^2 g_1 G_2 + 4 f_0^2 g_1^3}{h_1^3}, \\
    g_0&=0, \\
    a_1 &= \frac{g_1}{h_1} b_1, \\
    b_1^2 &= 2 \frac{b_0 e_1}{e_0} b_1 - \frac{4 b_0^3 h_1^2 + b_0^2 e_0 e_1^2 - b_0 e_0^2 h_1^2 g_2 + e_0^3 h_1^2 g_3}{e_0^3}, \\
    c_1&=0, \\
    d_1 &= \frac{e_1 g_1}{h_1}.
\end{align*}

\subsection{Case 2: $h_0 = 0$}
\subsubsection{Case 2.1: $h_0=0$, $ i_1\, e_0\,(e_1 i_1 - f_1 h_1)\not=0$}
\begin{align*}
    a_0 &= -\frac{i_0 h_1 c_1 d_1}{e_1 i_1^2 - f_1 i_1 h_1}  + \frac{i_0 e_1 g_1 c_1}{e_1 i_1^2 - f_1 i_1 h_1}  - \frac{ b_0\bigl(- e_0 i_1^2 +  f_0 i_1 h_1 -  i_0 f_1 h_1\bigr)}{e_0 \bigl(e_1 i_1^2 - f_1 i_1 h_1\bigr)} d_1 \\
    &\quad\qquad - \frac{b_0 g_1\bigl( e_0 f_1  i_1 - f_0 e_1  i_1 +  i_0 e_1 f_1 \bigr)}{e_0 \bigl(e_1 i_1^2 - f_1 i_1 h_1\bigr)}, \\
    c_0 &= \frac{i_0}{i_1} c_1 - \frac{-b_0 f_0 i_1 + b_0 i_0 f_1}{e_0 i_1}, \\
    d_0 &= \frac{e_0 i_1 - f_0 h_1}{e_1 i_1 - f_1 h_1} d_1 - \frac{e_0 f_1 g_1 - f_0 e_1 g_1}{e_1 i_1 - f_1 h_1}, \\
    g_0 &= -\frac{i_0 h_1}{e_1 i_1 - f_1 h_1} d_1 + \frac{i_0 e_1 g_1}{e_1 i_1 - f_1 h_1}, \\
    a_1 &= \frac{g_1}{i_1} c_1 + \frac{b_0}{e_0} d_1 - \frac{b_0 f_1 g_1}{e_0 i_1}, \\
    b_1 &= \frac{h_1}{i_1} c_1 - \frac{-b_0 e_1 i_1 + b_0 f_1 h_1}{e_0 i_1}, \\
    c_1^2 &= \frac{2 b_0 f_1}{e_0} c_1 - \frac{4 b_0^3 i_1^2 + b_0^2 e_0 f_1^2 - b_0 e_0^2 i_1^2 g_2 + e_0^3 i_1^2 g_3}{e_0^3}, \\
    d_1^3 &= -\frac{1}{i_1^3}\,\Biggl(\frac{1}{4} \bigl(e_1 i_1 h_1^2 - 12 f_1 g_1 i_1^2 -  f_1 h_1^3\bigr) d_1^2 \\
    &\qquad\quad + \bigl(-\frac{1}{2} e_1^2 g_1 i_1 h_1 - \frac{1}{4} e_1^2 i_1^3 G_2 + \frac{1}{2} e_1 f_1 g_1 h_1^2 + \frac{1}{2} e_1 f_1 i_1^2 h_1 G_2 \\
    &\qquad\qquad\quad + 3 f_1^2 g_1^2 i_1 - \frac{1}{4} f_1^2 i_1 h_1^2 G_2\bigr) d_1 \\
    &\qquad\quad + \Bigl(\frac{1}{4} e_1^3 g_1^2 i_1 + \frac{1}{4} e_1^3 i_1^3 G_3 - \frac{1}{4} e_1^2 f_1 g_1^2 h_1 + \frac{1}{4} e_1^2 f_1 g_1 i_1^2 G_2 \\
    &\qquad\qquad\quad - \frac{3}{4} e_1^2 f_1 i_1^2 h_1 G_3 - \frac{1}{2} e_1 f_1^2 g_1 i_1 h_1 G_2\\
    &\qquad\qquad\quad + \frac{3}{4} e_1 f_1^2 i_1 h_1^2 G_3 - f_1^3 g_1^3 + \frac{1}{4} f_1^3 g_1 h_1^2 G_2 - \frac{1}{4} f_1^3 h_1^3 G_3\Bigr)\Biggr).
\end{align*}
\subsubsection{Case 2.2: $h_0=0$, $(e_1 i_1 - f_1 h_1)=0$, $ i_1\,h_1\,i_0 e_0\not=0$}
\begin{align*}
    a_0 &= \frac{g_0 c_1}{i_1}  - \frac{b_0\bigl( e_0 i_1^2 -  f_0 h_1 i_1 +  i_0 f_1 h_1\bigr)}{e_0 i_0 h_1 i_1} g_0 + \frac{b_0 g_1}{h_1}, \\
    c_0 &= \frac{i_0}{i_1} c_1 + \frac{b_0\bigl(  f_0 i_1 -  i_0 f_1\bigr)}{e_0 i_1}, \\
    d_0 &= -\frac{e_0 i_1 - f_0 h_1}{i_0 h_1} g_0 + \frac{e_0 g_1}{h_1}, \\
    g_0^3 &= \frac{i_0\bigl(3  g_1 i_1^2 + \frac{1}{4}  h_1^3\bigr)}{i_1^3} g_0^2 - \frac{i_0^2 \bigl(3 g_1^2 - \frac{1}{4}  h_1^2 G_2\bigr)}{i_1^2} g_0 \\
    &\qquad\quad+ \frac{i_0^3 \bigl( g_1^3 - \frac{1}{4}  g_1 h_1^2 G_2 + \frac{1}{4}  h_1^3 G_3\bigr)}{i_1^3}, \\
    a_1 &= \frac{g_1}{i_1} c_1, \quad 
    b_1 = \frac{h_1}{i_1} c_1, \\
    c_1^2 &= \frac{2 b_0 f_1}{e_0} c_1 - \frac{4 b_0^3 i_1^2 + b_0^2 e_0 f_1^2 - b_0 e_0^2 i_1^2 g_2 + e_0^3 i_1^2 g_3}{e_0^3}, \\
    d_1 &= \frac{f_1 g_1}{i_1}, \quad
    e_1 = \frac{f_1 h_1}{i_1}.
\end{align*}
\subsubsection{Case 2.3: $h_0=0$, $(e_1 i_1 - f_1 h_1)=0$, $h_1=0$, $ i_1\, e_0\not=0$}

\begin{align*}
    a_0 &= \frac{b_0}{e_0} d_0 + \frac{i_0 g_1}{i_1^2} c_1 - \frac{b_0 i_0 f_1 g_1}{e_0 i_1^2}, \\
    c_0 &= \frac{i_0}{i_1} c_1 + \frac{b_0\bigl( f_0 i_1 -  i_0 f_1\bigr)}{e_0 i_1}, \\
    d_0^3 &= 3 \frac{f_0 g_1}{i_1} d_0^2 + \frac{\frac{1}{4} e_0^2 i_1^2 G_2 - 3 f_0^2 g_1^2}{i_1^2} d_0 - \frac{1}{4}\frac{ e_0^3 g_1^2 i_1 +  e_0^3 i_1^3 G_3 + e_0^2 f_0 g_1 i_1^2 G_2 - 4 f_0^3 g_1^3}{i_1^3}, \\
    g_0 &= \frac{i_0 g_1}{i_1}, \quad
    a_1 = \frac{g_1}{i_1} c_1, \quad
    b_1 = 0, \\
    c_1^2 &= \frac{2 b_0 f_1}{e_0} c_1 - \frac{4 b_0^3 i_1^2 + b_0^2 e_0 f_1^2 - b_0 e_0^2 i_1^2 g_2 + e_0^3 i_1^2 g_3}{e_0^3}, \quad
    d_1 = \frac{f_1 g_1}{i_1}, e_1=0.
\end{align*}
\subsubsection{Case 2.4: $h_0=0$, $(e_1 i_1 - f_1 h_1)=0$, $i_0=0$, $ i_1\, e_0\not=0$}
\begin{align*}
    a_0 &= \frac{b_0}{e_0} d_0, \quad
    c_0 = \frac{b_0 f_0}{e_0}, \\
    d_0^3 &= -\frac{1}{i_1^3}
    \Biggl(-\frac{1}{4} \bigl(e_0 h_1^2 i_1 - 12 f_0 g_1 i_1^2 - f_0 h_1^3\bigr) d_0^2 \\
    &\qquad\quad - \bigl(-\frac{1}{2} e_0^2 g_1 h_1 i_1 - \frac{1}{4} e_0^2 i_1^3 G_2 + \frac{1}{2} e_0 f_0 g_1 h_1^2 + \frac{1}{2} e_0 f_0 h_1 i_1^2 G_2 + 3 f_0^2 g_1^2 i_1 - \frac{1}{4} f_0^2 h_1^2 i_1 G_2\bigr) d_0 \\
    &\qquad\quad -\bigl(\frac{1}{4} e_0^3 g_1^2 i_1 + \frac{1}{4} e_0^3 i_1^3 G_3 - \frac{1}{4} e_0^2 f_0 g_1^2 h_1 + \frac{1}{4} e_0^2 f_0 g_1 i_1^2 G_2 - \frac{3}{4} e_0^2 f_0 h_1 i_1^2 G_3 \\
    &\qquad\qquad\quad - \frac{1}{2} e_0 f_0^2 g_1 h_1 i_1 G_2 + \frac{3}{4} e_0 f_0^2 h_1^2 i_1 G_3 - f_0^3 g_1^3 + \frac{1}{4} f_0^3 g_1 h_1^2 G_2 - \frac{1}{4} f_0^3 h_1^3 G_3\bigr)\Biggr), \\
    g_0&=0, \quad
    a_1 = \frac{g_1}{i_1} c_1, \quad
    b_1 = \frac{h_1}{i_1} c_1, \\
    c_1^2 &= \frac{2 b_0 f_1}{e_0} c_1 - \frac{4 b_0^3 i_1^2 + b_0^2 e_0 f_1^2 - b_0 e_0^2 i_1^2 g_2 + e_0^3 i_1^2 g_3}{e_0^3}, \\
    d_1 &= \frac{f_1 g_1}{i_1}, \quad
    e_1 = \frac{f_1 h_1}{i_1}.
\end{align*}

\subsection{Case $e_0 i_1 - h_0 f_1=0$}
\subsubsection{$e_0 i_1 - h_0 f_1=0$, $i_1(e_1 i_1 - f_1 h_1) h_0\not=0$}
\begin{align*}
    a_0 &= -\frac{(f_0  i_1 - i_0 f_1 )}{e_1^2 i_1^2 - 2 e_1 f_1 h_1 i_1 + f_1^2 h_1^2} b_1 d_1 h_1 + \frac{f_0 i_1 - i_0  f_1}{e_1^2 i_1^2 - 2 e_1 f_1 h_1 i_1 + f_1^2 h_1^2} b_1 e_1 g_1 \\
    &\quad - \frac{-b_0 f_0 h_1^2 i_1 - b_0 h_0 e_1 i_1^2 + b_0 h_0 f_1 h_1 i_1 + b_0 i_0 e_1 h_1 i_1}{h_0\bigl( e_1^2 i_1^2 - 2  e_1 f_1 h_1 i_1 +  f_1^2 h_1^2\bigr)} d_1 \\
    &\quad - \frac{b_0 f_0 e_1 g_1 h_1 i_1 + b_0 h_0 e_1 f_1 g_1 i_1 - b_0 h_0 f_1^2 g_1 h_1 - b_0 i_0 e_1^2 g_1 i_1}{h_0\bigl( e_1^2 i_1^2 - 2  e_1 f_1 h_1 i_1 +  f_1^2 h_1^2\bigr)}, \\
    c_0 &= \frac{f_0 i_1 - i_0 f_1}{e_1 i_1 - f_1 h_1} b_1 - \frac{b_0 f_0 h_1 i_1 - b_0 i_0 e_1 i_1}{h_0\bigl( e_1 i_1 -  f_1 h_1\bigr)}, \\
    d_0 &= -\frac{f_0 h_1 - h_0 f_1}{e_1 i_1 - f_1 h_1} d_1 - \frac{-g_1\bigl(f_0 e_1  i_1 + h_0 f_1^2 \bigr)}{i_1\bigl(e_1 i_1 - f_1 h_1\bigr)}, \\
    e_0 &= \frac{h_0 f_1}{i_1}, \\
    g_0 &= \frac{h_0 i_1 - i_0 h_1}{e_1 i_1 - f_1 h_1} d_1 + \frac{h_0 f_1 g_1 - i_0 e_1 g_1}{e_1 i_1 - f_1 h_1}, \\
    a_1 &= \frac{i_1}{e_1 i_1 - f_1 h_1} b_1 d_1 - \frac{f_1 g_1}{e_1 i_1 - f_1 h_1} b_1 - \frac{b_0 h_1 i_1}{h_0 \bigl(e_1 i_1 -  f_1 h_1\bigr)} d_1 + \frac{b_0 e_1 g_1 i_1}{h_0 \bigl(e_1 i_1 - f_1 h_1\bigr)}, \\
    b_1^3 &= \frac{1}{h_0 i_1^3}\bigl(3 b_0 h_1 i_1^3 - \frac{1}{4} h_0 e_1 f_1^2 i_1 + \frac{1}{4} h_0 f_1^3 h_1\bigr) b_1^2 \\
    & - \frac{1}{h_0^2 i_1^2}\bigl(3 b_0^2 h_1^2 i_1^2 - \frac{1}{2} b_0 h_0 e_1^2 f_1 i_1 + \frac{1}{2} b_0 h_0 e_1 f_1^2 h_1 - \frac{1}{4} h_0^2 e_1^2 i_1^2 g_2 \\
    &\qquad\quad + \frac{1}{2} h_0^2 e_1 f_1 h_1 i_1 g_2 - \frac{1}{4} h_0^2 f_1^2 h_1^2 g_2\bigr) b_1 \\
    & - \frac{1}{h_0^3 i_1^3}\bigl(-b_0^3 h_1^3 i_1^3 
    + \frac{1}{4} b_0^2 h_0 e_1^3 i_1^3 
    - \frac{1}{4} b_0^2 h_0 e_1^2 f_1 h_1 i_1^2 \\
    &\qquad\quad
    + \frac{1}{4} b_0 h_0^2 e_1^2 h_1 i_1^3 g_2 
    - \frac{1}{2} b_0 h_0^2 e_1 f_1 h_1^2 i_1^2 g_2 
    + \frac{1}{4} b_0 h_0^2 f_1^2 h_1^3 i_1 g_2 \\
    &\qquad\quad
    + \frac{1}{4} h_0^3 e_1^3 i_1^3 g_3 
    - \frac{3}{4} h_0^3 e_1^2 f_1 h_1 i_1^2 g_3 
    + \frac{3}{4} h_0^3 e_1 f_1^2 h_1^2 i_1 g_3 
    - \frac{1}{4} h_0^3 f_1^3 h_1^3 g_3\bigr), \\
    c_1 &= \frac{b_0 i_1}{h_0}, \\
    d_1^3 &= -\,\frac{1}{i_1^3}
    \Biggl(
    \frac{1}{4} \bigl(e_1 h_1^2 i_1 - 12 f_1 g_1 i_1^2 - f_1 h_1^3\bigr)d_1^2 \\
    &\qquad 
    + \bigl(-\frac{1}{2} e_1^2 g_1 h_1 i_1 - \frac{1}{4} e_1^2 i_1^3 G_2 + \frac{1}{2} e_1 f_1 g_1 h_1^2 + \frac{1}{2} e_1 f_1 h_1 i_1^2 G_2 + 3 f_1^2 g_1^2 i_1 - \frac{1}{4} f_1^2 h_1^2 i_1 G_2\bigr) d_1 \\
    &\qquad + \bigl(\frac{1}{4} e_1^3 g_1^2 i_1 + \frac{1}{4} e_1^3 i_1^3 G_3 - \frac{1}{4} e_1^2 f_1 g_1^2 h_1 + \frac{1}{4} e_1^2 f_1 g_1 i_1^2 G_2 - \frac{3}{4} e_1^2 f_1 h_1 i_1^2 G_3 \\
    &\qquad\qquad - \frac{1}{2} e_1 f_1^2 g_1 h_1 i_1 G_2 + \frac{3}{4} e_1 f_1^2 h_1^2 i_1 G_3 - f_1^3 g_1^3 + \frac{1}{4} f_1^3 g_1 h_1^2 G_2 - \frac{1}{4} f_1^3 h_1^3 G_3\bigr)\Biggr).
\end{align*}
\subsubsection{$e_0 i_1 - h_0 f_1=0$, $e_1 i_1 - f_1 h_1=0$, $i_1 h_0(h_0i_1-i_0h_1)\not=0$}
\begin{align*}
    a_0 &= -\frac{(f_0 h_1 i_1 - i_0 f_1 h_1)}{\bigl(e_1 i_1 - f_1 h_1\bigr)^2} b_1 d_1 
    - \frac{(-f_0 e_1 g_1 i_1 + i_0 e_1 f_1 g_1)}{\bigl(e_1 i_1 - f_1 h_1\bigr)^2} b_1 \\
    &\quad - \frac{(-b_0 f_0 h_1^2 i_1 - b_0 h_0 e_1 i_1^2 + b_0 h_0 f_1 h_1 i_1 + b_0 i_0 e_1 h_1 i_1)}{h_0\bigl(e_1 i_1 - f_1 h_1\bigr)^2} d_1 \\
    &\quad - \frac{(b_0 f_0 e_1 g_1 h_1 i_1 + b_0 h_0 e_1 f_1 g_1 i_1 - b_0 h_0 f_1^2 g_1 h_1 - b_0 i_0 e_1^2 g_1 i_1)}{h_0\bigl(e_1 i_1 - f_1 h_1\bigr)^2}, \\
    c_0 &= \frac{(f_0 i_1 - i_0 f_1)}{e_1 i_1 - f_1 h_1} b_1 - \frac{(b_0 f_0 h_1 i_1 - b_0 i_0 e_1 i_1)}{h_0 \bigl(e_1 i_1 -  f_1 h_1\bigr)}, \\
    d_0 &= -\frac{(f_0 h_1 - h_0 f_1)}{e_1 i_1 - f_1 h_1} d_1 + \frac{g_1\bigl(f_0 e_1  i_1 - h_0 f_1^2 \bigr)}{i_1 \bigl(e_1 i_1 - f_1 h_1 \bigr)}, \\
    e_0 &= \frac{h_0 f_1}{i_1}, \\
    g_0 &= \frac{(h_0 i_1 - i_0 h_1)}{e_1 i_1 - f_1 h_1} d_1 + \frac{(h_0 f_1 g_1 - i_0 e_1 g_1)}{e_1 i_1 - f_1 h_1}, \\
    a_1 &= \frac{i_1}{e_1 i_1 - f_1 h_1} b_1 d_1 - \frac{f_1 g_1}{e_1 i_1 - f_1 h_1} b_1 - \frac{b_0 h_1 i_1}{h_0 \bigl(e_1 i_1 - f_1 h_1\bigr)} d_1 + \frac{b_0 e_1 g_1 i_1}{h_0 \bigl(e_1 i_1 - f_1 h_1\bigr)}, \\
    b_1^3 &= \frac{3 b_0 h_1 i_1^3 - \frac{1}{4} h_0 e_1 f_1^2 i_1 + \frac{1}{4} h_0 f_1^3 h_1}{h_0 i_1^3} b_1^2 \\
    &\quad - \frac{3 b_0^2 h_1^2 i_1^2 - \frac{1}{2} b_0 h_0 e_1^2 f_1 i_1 + \frac{1}{2} b_0 h_0 e_1 f_1^2 h_1 - \frac{1}{4} h_0^2 e_1^2 i_1^2 g_2 + \frac{1}{2} h_0^2 e_1 f_1 h_1 i_1 g_2 - \frac{1}{4} h_0^2 f_1^2 h_1^2 g_2}{h_0^2 i_1^2} b_1 \\
    &\quad - \frac{1}{h_0^3 i_1^3}\bigl(-b_0^3 h_1^3 i_1^3 + \frac{1}{4} b_0^2 h_0 e_1^3 i_1^3 - \frac{1}{4} b_0^2 h_0 e_1^2 f_1 h_1 i_1^2 + \frac{1}{4} b_0 h_0^2 e_1^2 h_1 i_1^3 g_2 \\
    &\qquad\qquad - \frac{1}{2} b_0 h_0^2 e_1 f_1 h_1^2 i_1^2 g_2 + \frac{1}{4} b_0 h_0^2 f_1^2 h_1^3 i_1 g_2 + \frac{1}{4} h_0^3 e_1^3 i_1^3 g_3 \\
    &\qquad\qquad - \frac{3}{4} h_0^3 e_1^2 f_1 h_1 i_1^2 g_3 + \frac{3}{4} h_0^3 e_1 f_1^2 h_1^2 i_1 g_3 - \frac{1}{4} h_0^3 f_1^3 h_1^3 g_3\bigr), \\
    c_1 &= \frac{b_0 i_1}{h_0}, \\
    d_1^3 &= \frac{1}{i_1^3}
    \Biggl(
    -\frac{1}{4} \bigl(
    e_1 h_1^2 i_1 - 12 f_1 g_1 i_1^2 - \frac{1}{4} f_1 h_1^3
    \bigr) d_1^2 \\
    &\quad - \bigl(-\frac{1}{2} e_1^2 g_1 h_1 i_1 - \frac{1}{4} e_1^2 i_1^3 G_2 + \frac{1}{2} e_1 f_1 g_1 h_1^2 + \frac{1}{2} e_1 f_1 h_1 i_1^2 G_2 + 3 f_1^2 g_1^2 i_1 - \frac{1}{4} f_1^2 h_1^2 i_1 G_2\bigr) d_1 \\
    &\quad 
    - \bigl(
    \frac{1}{4} e_1^3 g_1^2 i_1 + \frac{1}{4} e_1^3 i_1^3 G_3 - \frac{1}{4} e_1^2 f_1 g_1^2 h_1 + \frac{1}{4} e_1^2 f_1 g_1 i_1^2 G_2 - \frac{3}{4} e_1^2 f_1 h_1 i_1^2 G_3 \\
    &\qquad\qquad - \frac{1}{2} e_1 f_1^2 g_1 h_1 i_1 G_2 + \frac{3}{4} e_1 f_1^2 h_1^2 i_1 G_3 - f_1^3 g_1^3 + \frac{1}{4} f_1^3 g_1 h_1^2 G_2 - \frac{1}{4} f_1^3 h_1^3 G_3\bigr)\Biggr).
\end{align*}
\subsubsection{$e_0 i_1 - h_0 f_1=0$, $e_1 i_1 - f_1 h_1=0$, 
$h_0 i_1 - i_0 h_1=0$, $i_1h_1i_0(f_0i_1-i_0f_1)\not=0$}
\begin{align*}
    a_0 &= -\frac{(f_0 h_1 i_1 - i_0 f_1 h_1)}{e_1^2 i_1^2 - 2 e_1 f_1 h_1 i_1 + f_1^2 h_1^2} b_1 d_1 - \frac{-f_0 e_1 g_1 i_1 + i_0 e_1 f_1 g_1}{e_1^2 i_1^2 - 2 e_1 f_1 h_1 i_1 + f_1^2 h_1^2} b_1 \\
    &\quad - \frac{-b_0 f_0 h_1^2 i_1 - b_0 h_0 e_1 i_1^2 + b_0 h_0 f_1 h_1 i_1 + b_0 i_0 e_1 h_1 i_1}{h_0 e_1^2 i_1^2 - 2 h_0 e_1 f_1 h_1 i_1 + h_0 f_1^2 h_1^2} d_1 \\
    &\quad - \frac{b_0 f_0 e_1 g_1 h_1 i_1 + b_0 h_0 e_1 f_1 g_1 i_1 - b_0 h_0 f_1^2 g_1 h_1 - b_0 i_0 e_1^2 g_1 i_1}{h_0 e_1^2 i_1^2 - 2 h_0 e_1 f_1 h_1 i_1 + h_0 f_1^2 h_1^2}, \\
    c_0 &= \frac{f_0 i_1 - i_0 f_1}{e_1 i_1 - f_1 h_1} b_1 - \frac{b_0 f_0 h_1 i_1 - b_0 i_0 e_1 i_1}{h_0 e_1 i_1 - h_0 f_1 h_1}, \\
    d_0 &= -\frac{(f_0 h_1 - h_0 f_1)}{e_1 i_1 - f_1 h_1} d_1 - \frac{-f_0 e_1 g_1 i_1 + h_0 f_1^2 g_1}{e_1 i_1^2 - f_1 h_1 i_1}, \\
    e_0 &= \frac{h_0 f_1}{i_1}, \\
    g_0 &= \frac{h_0 i_1 - i_0 h_1}{e_1 i_1 - f_1 h_1} d_1 - \frac{h_0 f_1 g_1 - i_0 e_1 g_1}{e_1 i_1 - f_1 h_1}, \\
    a_1 &= \frac{i_1}{e_1 i_1 - f_1 h_1} b_1 d_1 - \frac{f_1 g_1}{e_1 i_1 - f_1 h_1} b_1 - \frac{b_0 h_1 i_1}{h_0 e_1 i_1 - h_0 f_1 h_1} d_1 + \frac{b_0 e_1 g_1 i_1}{h_0 e_1 i_1 - h_0 f_1 h_1}, \\
    b_1^3 &= \frac{b_1^2}{h_0 i_1^3}  \bigl(3 b_0 h_1 i_1^3 - \frac{1}{4} h_0 e_1 f_1^2 i_1 + \frac{1}{4} h_0 f_1^3 h_1\bigr) \\
    &\quad - \frac{b_1}{h_0^2 i_1^2} \bigl(3 b_0^2 h_1^2 i_1^2 - \frac{1}{2} b_0 h_0 e_1^2 f_1 i_1 + \frac{1}{2} b_0 h_0 e_1 f_1^2 h_1 - \frac{1}{4} h_0^2 e_1^2 i_1^2 g_2 + \frac{1}{2} h_0^2 e_1 f_1 h_1 i_1 g_2 - \frac{1}{4} h_0^2 f_1^2 h_1^2 g_2\bigr) \\
    &\quad -\,\frac{1}{h_0^3 i_1^3}
    \bigl(-b_0^3 h_1^3 i_1^3 + \frac{1}{4} b_0^2 h_0 e_1^3 i_1^3 - \frac{1}{4} b_0^2 h_0 e_1^2 f_1 h_1 i_1^2 + \frac{1}{4} b_0 h_0^2 e_1^2 h_1 i_1^3 g_2 - \frac{1}{2} b_0 h_0^2 e_1 f_1 h_1^2 i_1^2 g_2 \\
    &\qquad\qquad + \frac{1}{4} b_0 h_0^2 f_1^2 h_1^3 i_1 g_2 + \frac{1}{4} h_0^3 e_1^3 i_1^3 g_3 - \frac{3}{4} h_0^3 e_1^2 f_1 h_1 i_1^2 g_3 + \frac{3}{4} h_0^3 e_1 f_1^2 h_1^2 i_1 g_3 - \frac{1}{4} h_0^3 f_1^3 h_1^3 g_3\bigr), \\
    c_1 &= \frac{b_0 i_1}{h_0}, \\
    d_1^3 &= -\frac{d_1^2}{4i_1^3}   \bigl(e_1 h_1^2 i_1 - 12 f_1 g_1 i_1^2 - \frac{1}{4} f_1 h_1^3\bigr)\\
    &\quad - \frac{d_1}{i_1^3} \bigl(-\frac{1}{2} e_1^2 g_1 h_1 i_1 - \frac{1}{4} e_1^2 i_1^3 G_2 + \frac{1}{2} e_1 f_1 g_1 h_1^2 + \frac{1}{2} e_1 f_1 h_1 i_1^2 G_2 + 3 f_1^2 g_1^2 i_1 - \frac{1}{4} f_1^2 h_1^2 i_1 G_2\bigr) \\
    &\quad - \frac{1}{i_1^3}\bigl(\frac{1}{4} e_1^3 g_1^2 i_1 + \frac{1}{4} e_1^3 i_1^3 G_3 - \frac{1}{4} e_1^2 f_1 g_1^2 h_1 + \frac{1}{4} e_1^2 f_1 g_1 i_1^2 G_2 - \frac{3}{4} e_1^2 f_1 h_1 i_1^2 G_3 - \frac{1}{2} e_1 f_1^2 g_1 h_1 i_1 G_2 \\
    &\qquad\qquad + \frac{3}{4} e_1 f_1^2 h_1^2 i_1 G_3 - f_1^3 g_1^3 + \frac{1}{4} f_1^3 g_1 h_1^2 G_2 - \frac{1}{4} f_1^3 h_1^3 G_3\bigr).
\end{align*}

\subsubsection{$e_0 i_1 - h_0 f_1=0$, $e_1 i_1 - f_1 h_1=0$, 
$h_0 i_1 - i_0 h_1=0$, $f_0 i_1 - i_0 f_1=0$}
This is a degenerate branch resulting in a set of cross relations.
\begin{align*}
     f_0 h_0 - e_0 i_0 = 0, \quad
    & b_0 a_1 - a_0 b_1 = 0, \quad
    & c_0 a_1 - a_0 c_1 = 0, \\
     d_0 a_1 - a_0 d_1 = 0, \quad
    & e_0 a_1 - a_0 e_1 = 0, \quad
    & f_0 a_1 - a_0 f_1 = 0, \\
     g_0 a_1 - a_0 g_1 = 0, \quad
    & h_0 a_1 - a_0 h_1 = 0, \quad
    & i_0 a_1 - a_0 i_1 = 0, \\
     c_0 b_1 - b_0 c_1 = 0, \quad
    & d_0 b_1 - b_0 d_1 = 0, \quad
    & e_0 b_1 - b_0 e_1 = 0, \\
     f_0 b_1 - b_0 f_1 = 0, \quad
    & g_0 b_1 - b_0 g_1 = 0, \quad
    & h_0 b_1 - b_0 h_1 = 0, \\
     i_0 b_1 - b_0 i_1 = 0, \quad
    & d_0 c_1 - c_0 d_1 = 0, \quad
    & e_0 c_1 - c_0 e_1 = 0, \\
     f_0 c_1 - c_0 f_1 = 0, \quad
    & g_0 c_1 - c_0 g_1 = 0, \quad
    & h_0 c_1 - c_0 h_1 = 0, \\
     i_0 c_1 - c_0 i_1 = 0, \quad
    & e_0 d_1 - d_0 e_1 = 0, \quad
    & f_0 d_1 - d_0 f_1 = 0, \\
     g_0 d_1 - d_0 g_1 = 0, \quad
    & h_0 d_1 - d_0 h_1 = 0, \quad
    & i_0 d_1 - d_0 i_1 = 0, \\
     f_0 e_1 - e_0 f_1 = 0, \quad
    & g_0 e_1 - e_0 g_1 = 0, \quad
    & h_0 e_1 - e_0 h_1 = 0, \\
     i_0 e_1 - e_0 i_1 = 0, \quad
    & g_0 f_1 - f_0 g_1 = 0, \quad
    & h_0 f_1 - e_0 i_1 = 0, \\
     i_0 f_1 - f_0 i_1 = 0, \quad
    & h_0 g_1 - g_0 h_1 = 0, \quad
    & i_0 g_1 - g_0 i_1 = 0, \\
     f_0 h_1 - e_0 i_1 = 0, \quad
    & i_0 h_1 - h_0 i_1 = 0, \quad
    & f_1 h_1 - e_1 i_1 = 0.
\end{align*}
\subsection{Case $e_1 i_1 - f_1 h_1=0$}
\subsubsection{$e_1 i_1 - f_1 h_1=0$, $i_1 h_0 (h_0 i_1 - i_0 h_1) (e_0 i_1 - h_0 f_1)\not=0$}
\begin{align*}
    a_0 &= -\frac{(e_0 i_0 h_1 - f_0 h_0 h_1)}{e_0 h_0 i_1^2 - e_0 i_0 h_1 i_1 - h_0^2 f_1 i_1 + h_0 i_0 f_1 h_1} g_0 c_1 \\
    &\quad - \frac{-b_0 e_0 i_1^2 + b_0 f_0 h_1 i_1 + b_0 h_0 f_1 i_1 - b_0 i_0 f_1 h_1}{e_0 h_0 i_1^2 - e_0 i_0 h_1 i_1 - h_0^2 f_1 i_1 + h_0 i_0 f_1 h_1} g_0 \\
    &\quad - \frac{-e_0 h_0 i_0 g_1 + f_0 h_0^2 g_1}{e_0 h_0 i_1^2 - e_0 i_0 h_1 i_1 - h_0^2 f_1 i_1 + h_0 i_0 f_1 h_1} c_1 \\
    &\quad - \frac{b_0 e_0 i_0 g_1 i_1 - b_0 f_0 h_0 g_1 i_1}{e_0 h_0 i_1^2 - e_0 i_0 h_1 i_1 - h_0^2 f_1 i_1 + h_0 i_0 f_1 h_1}, \\
    c_0 &= \frac{e_0 i_0 - f_0 h_0}{e_0 i_1 - h_0 f_1} c_1 - \frac{-b_0 f_0 i_1 + b_0 i_0 f_1}{e_0 i_1 - h_0 f_1}, \\
    d_0 &= \frac{e_0 i_1 - f_0 h_1}{h_0 i_1 - i_0 h_1} g_0 - \frac{e_0 i_0 g_1 - f_0 h_0 g_1}{h_0 i_1 - i_0 h_1}, \\
    g_0^3 &= -\frac{g_0^2}{4i_1^3}  \bigl(h_0 h_1^2 i_1 - 12 i_0 g_1 i_1^2 - i_0 h_1^3\bigr) \\
    &\quad - \frac{g_0}{i_1^3}\bigl(-\frac{1}{2} h_0^2 g_1 h_1 i_1 - \frac{1}{4} h_0^2 i_1^3 G_2 + \frac{1}{2} h_0 i_0 g_1 h_1^2 + \frac{1}{2} h_0 i_0 h_1 i_1^2 G_2 + 3 i_0^2 g_1^2 i_1 - \frac{1}{4} i_0^2 h_1^2 i_1 G_2
    \bigr) \\
    &\quad - \frac{1}{i_1^3}\bigl(\frac{1}{4} h_0^3 g_1^2 i_1 + \frac{1}{4} h_0^3 i_1^3 G_3 - \frac{1}{4} h_0^2 i_0 g_1^2 h_1 + \frac{1}{4} h_0^2 i_0 g_1 i_1^2 G_2 - \frac{3}{4} h_0^2 i_0 h_1 i_1^2 G_3 - \frac{1}{2} h_0 i_0^2 g_1 h_1 i_1 G_2 \\
    &\qquad\qquad + \frac{3}{4} h_0 i_0^2 h_1^2 i_1 G_3 - i_0^3 g_1^3 + \frac{1}{4} i_0^3 g_1 h_1^2 G_2 - \frac{1}{4} i_0^3 h_1^3 G_3\bigr), \\
    a_1 &= \frac{g_1}{i_1} c_1, \quad 
    b_1 = \frac{h_1}{i_1} c_1, \\
    c_1^3 &= \frac{c_1^2}{h_0^3}\bigl(3 b_0 h_0^2 i_1 + \frac{1}{4} e_0^3 i_1 - \frac{1}{4} e_0^2 h_0 f_1\bigr) \\
    &\quad - \frac{c_1}{h_0^3}\bigl(3 b_0^2 h_0 i_1^2 + \frac{1}{2} b_0 e_0^2 f_1 i_1 - \frac{1}{2} b_0 e_0 h_0 f_1^2 - \frac{1}{4} e_0^2 h_0 i_1^2 g_2 + \frac{1}{2} e_0 h_0^2 f_1 i_1 g_2 - \frac{1}{4} h_0^3 f_1^2 g_2\bigr) \\
    &\quad - \frac{1}{h_0^3}\bigl(-b_0^3 i_1^3 - \frac{1}{4} b_0^2 e_0 f_1^2 i_1 + \frac{1}{4} b_0^2 h_0 f_1^3 + \frac{1}{4} b_0 e_0^2 i_1^3 g_2 - \frac{1}{2} b_0 e_0 h_0 f_1 i_1^2 g_2 + \frac{1}{4} b_0 h_0^2 f_1^2 i_1 g_2 \\
    &\qquad\qquad - \frac{1}{4} e_0^3 i_1^3 g_3 + \frac{3}{4} e_0^2 h_0 f_1 i_1^2 g_3 - \frac{3}{4} e_0 h_0^2 f_1^2 i_1 g_3 + \frac{1}{4} h_0^3 f_1^3 g_3\bigr), \\
    d_1 &= \frac{f_1 g_1}{i_1}, \quad
    e_1 = \frac{f_1 h_1}{i_1}.
\end{align*}
\subsubsection{$e_1 i_1 - f_1 h_1=0$, $h_0 i_1 - i_0 h_1=0$, $i_1 h_1 i_0 (e_0 i_1^2 - i_0 f_1 h_1)\not=0$}
\begin{align*}
    a_0 &= \frac{1}{e_0 i_1^2 - i_0 f_1 h_1}\bigl(i_0 h_1 d_0 c_1 + b_0 i_1^2 d_0 + e_0 i_0 g_1 c_1 - b_0 i_0 f_1 g_1\bigr), \\
    c_0 &= \frac{1}{e_0 i_1^2 - i_0 f_1 h_1}\bigl( (e_0 i_0 i_1 - f_0 i_0 h_1) c_1 - (-b_0 f_0 i_1^2 + b_0 i_0 f_1 i_1)\bigr), \\
    d_0^3 &= -\frac{d_0^2}{4i_1^3} \bigl(e_0 h_1^2 i_1 - 12 f_0 g_1 i_1^2 - f_0 h_1^3\bigr)  \\
    &\quad -\frac{d_0}{i_1^3} \bigl( -\frac{1}{2} e_0^2 g_1 h_1 i_1 - \frac{1}{4} e_0^2 i_1^3 G_2 + \frac{1}{2} e_0 f_0 g_1 h_1^2 + \frac{1}{2} e_0 f_0 h_1 i_1^2 G_2 + 3 f_0^2 g_1^2 i_1 - \frac{1}{4} f_0^2 h_1^2 i_1 G_2\bigr) \\
    &\quad -\frac{1}{i_1^3}\bigl(\frac{1}{4} e_0^3 g_1^2 i_1 + \frac{1}{4} e_0^3 i_1^3 G_3 - \frac{1}{4} e_0^2 f_0 g_1^2 h_1 + \frac{1}{4} e_0^2 f_0 g_1 i_1^2 G_2 - \frac{3}{4} e_0^2 f_0 h_1 i_1^2 G_3 \\
    &\qquad\qquad - \frac{1}{2} e_0 f_0^2 g_1 h_1 i_1 G_2 + \frac{3}{4} e_0 f_0^2 h_1^2 i_1 G_3 - f_0^3 g_1^3 + \frac{1}{4} f_0^3 g_1 h_1^2 G_2 - \frac{1}{4} f_0^3 h_1^3 G_3\bigr), \\
    g_0 &= \frac{i_0 g_1}{i_1}, \quad
    h_0 = \frac{i_0 h_1}{i_1}, \\
    a_1 &= \frac{g_1}{i_1} c_1, \quad
    b_1 = \frac{h_1}{i_1} c_1, \\
    c_1^3 &= \frac{c_1^2}{i_0^3 h_1^3} \bigl(3 b_0 i_0^2 h_1^2 i_1^2 + \frac{1}{4} e_0^3 i_1^4 - \frac{1}{4} e_0^2 i_0 f_1 h_1 i_1^2\bigr) \\
    &\quad - \frac{c_1}{i_0^3 h_1^3} \bigl(3 b_0^2 i_0 h_1 i_1^4 + \frac{1}{2} b_0 e_0^2 f_1 i_1^4 - \frac{1}{2} b_0 e_0 i_0 f_1^2 h_1 i_1^2 - \frac{1}{4} e_0^2 i_0 h_1 i_1^4 g_2 + \frac{1}{2} e_0 i_0^2 f_1 h_1^2 i_1^2 g_2 - \frac{1}{4} i_0^3 f_1^2 h_1^3 g_2
    \bigr) \\
    &\quad - \frac{1}{i_0^3 h_1^3} \bigl(-b_0^3 i_1^6 - \frac{1}{4} b_0^2 e_0 f_1^2 i_1^4 + \frac{1}{4} b_0^2 i_0 f_1^3 h_1 i_1^2 + \frac{1}{4} b_0 e_0^2 i_1^6 g_2 - \frac{1}{2} b_0 e_0 i_0 f_1 h_1 i_1^4 g_2\\
    &\qquad\qquad + \frac{1}{4} b_0 i_0^2 f_1^2 h_1^2 i_1^2 g_2 - \frac{1}{4} e_0^3 i_1^6 g_3 + \frac{3}{4} e_0^2 i_0 f_1 h_1 i_1^4 g_3 - \frac{3}{4} e_0 i_0^2 f_1^2 h_1^2 i_1^2 g_3 + \frac{1}{4} i_0^3 f_1^3 h_1^3 g_3\bigr), \\
    d_1 &= \frac{f_1 g_1}{i_1}, \quad 
    e_1 = \frac{f_1 h_1}{i_1}.
\end{align*}

\subsubsection{$e_1 i_1 - f_1 h_1=0$, $h_0 i_1 - i_0 h_1=0$, $h_1=0$, $i_1 h_0 \not=0$}
\begin{align*}
    a_0 &= \frac{b_0}{e_0} d_0 + \frac{i_0 g_1}{i_1^2} c_1 - \frac{b_0 i_0 f_1 g_1}{e_0 i_1^2}, \quad 
    c_0 = \frac{i_0}{i_1} c_1 - \frac{-b_0 f_0 i_1 + b_0 i_0 f_1}{e_0 i_1}, \\
    d_0^3 &= 3 \frac{f_0 g_1}{i_1} d_0^2 - \frac{-\frac{1}{4} e_0^2 i_1^2 G_2 + 3 f_0^2 g_1^2}{i_1^2} d_0 \\
    &\quad - \frac{\frac{1}{4} e_0^3 g_1^2 i_1 + \frac{1}{4} e_0^3 i_1^3 G_3 + \frac{1}{4} e_0^2 f_0 g_1 i_1^2 G_2 - f_0^3 g_1^3}{i_1^3}, \\
    g_0 &= \frac{i_0 g_1}{i_1}, \quad
    h_0=0, \quad
    a_1 = \frac{g_1}{i_1} c_1, \quad
    b_1=0, \\
    c_1^2 &= \frac{2 b_0 f_1}{e_0} c_1 - \frac{4 b_0^3 i_1^2 + b_0^2 e_0 f_1^2 - b_0 e_0^2 i_1^2 g_2 + e_0^3 i_1^2 g_3}{e_0^3}, \\
    d_1 &= \frac{f_1 g_1}{i_1}, \quad
    e_1=0, \quad
    h_1=0.
\end{align*}
\subsubsection{$e_1 i_1 - f_1 h_1=0$, $h_0 i_1 - i_0 h_1=0$, $i_0=0$, $i_1 e_0 \not=0$}
\begin{align*}
    a_0 &= \frac{b_0}{e_0} d_0, \quad 
    c_0 = \frac{b_0 f_0}{e_0}, \\
    d_0^3 &= -\frac{d_0^2}{4 i_1^3} \bigl(e_0 i_1 h_1^2 - 12 f_0 g_1 i_1^2 - f_0 h_1^3\bigr)  \\
    &\quad - \frac{d_0}{i_1^3} \bigl(-\frac{1}{2} e_0^2 g_1 i_1 h_1 - \frac{1}{4} e_0^2 i_1^3 G_2 + \frac{1}{2} e_0 f_0 g_1 h_1^2 + \frac{1}{2} e_0 f_0 i_1^2 h_1 G_2 + 3 f_0^2 g_1^2 i_1 - \frac{1}{4} f_0^2 i_1 h_1^2 G_2\bigr) \\
    &\quad - \frac{1}{i_1^3} \bigl(\frac{1}{4} e_0^3 g_1^2 i_1 + \frac{1}{4} e_0^3 i_1^3 G_3 - \frac{1}{4} e_0^2 f_0 g_1^2 h_1 + \frac{1}{4} e_0^2 f_0 g_1 i_1^2 G_2 - \frac{3}{4} e_0^2 f_0 i_1^2 h_1 G_3 - \frac{1}{2} e_0 f_0^2 g_1 i_1 h_1 G_2 \\
    &\qquad\qquad + \frac{3}{4} e_0 f_0^2 i_1 h_1^2 G_3 - f_0^3 g_1^3 + \frac{1}{4} f_0^3 g_1 h_1^2 G_2 - \frac{1}{4} f_0^3 h_1^3 G_3\bigr), \\
    g_0&=0, \quad 
    h_0, \\
    a_1 &= \frac{g_1}{i_1} c_1, \quad 
    b_1 = \frac{h_1}{i_1} c_1, \\
    c_1^2 &= \frac{2 b_0 f_1}{e_0} c_1 - \frac{4 b_0^3 i_1^2 + b_0^2 e_0 f_1^2 - b_0 e_0^2 i_1^2 g_2 + e_0^3 i_1^2 g_3}{e_0^3}, \\
    d_1 &= \frac{f_1 g_1}{i_1}, \quad
    e_1 = \frac{f_1 h_1}{i_1}.
\end{align*}
\subsubsection{$e_1 i_1 - f_1 h_1=0$, $h_0 i_1 - i_0 h_1=0$, $e_0 i_1^2 - i_0 f_1 h_1=0$, $i_1 h_1 i_0 (f_0 i_1 - i_0 f_1) \not=0$}
\begin{align*}
    a_0 &= \frac{1}{f_0 i_1 - i_0 f_1}\bigl(i_1 c_0 d_0 - i_0 f_1 g_1 c_0 - b_0 i_1^2 d_0 + b_0 f_0 g_1 i_1\bigr), \\
    c_0^3 &= \frac{c_0^2}{h_1 i_1^3}\bigl( 3 b_0 i_1^4 - \frac{1}{4} f_0 f_1^2 h_1 i_1 + \frac{1}{4} i_0 f_1^3 h_1\bigr) \\
    &\quad - \frac{c_0}{i_0 h_1^2 i_1^2}\bigl(3 b_0^2 i_0 i_1^4 - \frac{1}{2} b_0 f_0^2 f_1 h_1 i_1^2 + \frac{1}{2} b_0 f_0 i_0 f_1^2 h_1 i_1 - \frac{1}{4} f_0^2 i_0 h_1^2 i_1^2 g_2 + \frac{1}{2} f_0 i_0^2 f_1 h_1^2 i_1 g_2 - \frac{1}{4} i_0^3 f_1^2 h_1^2 g_2\bigr) \\
    &\quad - \frac{1}{i_0^2 h_1^3 i_1^3}\bigl(-b_0^3 i_0^2 i_1^6 + \frac{1}{4} b_0^2 f_0^3 h_1 i_1^5 - \frac{1}{4} b_0^2 f_0^2 i_0 f_1 h_1 i_1^4 + \frac{1}{4} b_0 f_0^2 i_0^2 h_1^2 i_1^4 g_2 - \frac{1}{2} b_0 f_0 i_0^3 f_1 h_1^2 i_1^3 g_2 \\
    &\qquad\qquad\quad + \frac{1}{4} b_0 i_0^4 f_1^2 h_1^2 i_1^2 g_2 + \frac{1}{4} f_0^3 i_0^2 h_1^3 i_1^3 g_3 - \frac{3}{4} f_0^2 i_0^3 f_1 h_1^3 i_1^2 g_3 + \frac{3}{4} f_0 i_0^4 f_1^2 h_1^3 i_1 g_3 - \frac{1}{4} i_0^5 f_1^3 h_1^3 g_3\bigr), \\
    d_0^3 &= -\frac{d_0^2}{i_1^4} \bigl(3 f_0 g_1 i_1^3 + \frac{1}{4} f_0 h_1^3 i_1 - \frac{1}{4} i_0 f_1 h_1^3\bigr) \\
    &\quad -\frac{d_0}{i_1^6} \bigl(3 f_0^2 g_1^2 i_1^4 - \frac{1}{4} f_0^2 h_1^2 i_1^4 G_2 + \frac{1}{2} f_0 i_0 f_1 g_1 h_1^3 i_1 + \frac{1}{2} f_0 i_0 f_1 h_1^2 i_1^3 G_2 - \frac{1}{2} i_0^2 f_1^2 g_1 h_1^3 - \frac{1}{4} i_0^2 f_1^2 h_1^2 i_1^2 G_2\bigr)\\
    &\quad - \frac{1}{i_1^8} \bigl(-f_0^3 g_1^3 i_1^5 + \frac{1}{4} f_0^3 g_1 h_1^2 i_1^5 G_2 - \frac{1}{4} f_0^3 h_1^3 i_1^5 G_3 - \frac{1}{2} f_0^2 i_0 f_1 g_1 h_1^2 i_1^4 G_2 + \frac{3}{4} f_0^2 i_0 f_1 h_1^3 i_1^4 G_3 \\
    &\qquad\qquad\quad - \frac{1}{4} f_0 i_0^2 f_1^2 g_1^2 h_1^3 i_1 + \frac{1}{4} f_0 i_0^2 f_1^2 g_1 h_1^2 i_1^3 G_2 - \frac{3}{4} f_0 i_0^2 f_1^2 h_1^3 i_1^3 G_3 + \frac{1}{4} i_0^3 f_1^3 g_1^2 h_1^3 + \frac{1}{4} i_0^3 f_1^3 h_1^3 i_1^2 G_3\bigr), \\
    e_0 &= \frac{i_0 f_1 h_1}{i_1^2}, \quad
    g_0 = \frac{i_0 g_1}{i_1}, \quad
    h_0 = \frac{i_0 h_1}{i_1}, \quad
    a_1 = \frac{b_0 g_1 i_1}{i_0 h_1}, \\
    b_1 &= \frac{b_0 i_1}{i_0}, \quad
    c_1 = \frac{b_0 i_1^2}{i_0 h_1}, \quad
    d_1 = \frac{f_1 g_1}{i_1}, \quad
    e_1 = \frac{f_1 h_1}{i_1}.
\end{align*}

\end{document}